\documentclass[conference,a4paper]{IEEEtran}

\usepackage{amsmath}
\usepackage{cases}
\usepackage{amsfonts}
\usepackage{amssymb}
\usepackage{boxedminipage}
\usepackage{graphicx}
\usepackage[usenames]{color}
\usepackage{array,color}
\usepackage{colortbl}
\usepackage{extarrows}
\usepackage{bm}
\usepackage{amsmath}
\usepackage[numbers,sort&compress]{natbib}

\newtheorem{thm}{Theorem}
\newtheorem{cor}{Corollary}

\newtheorem{defn}{Definition}
\newtheorem{rem}{Remark}

\begin{document}
\title{An Extended Fano's Inequality for the Finite Blocklength Coding}

\author{ Yunquan Dong, Pingyi Fan\\
\{dongyq08@mails,fpy@mail\}.tsinghua.edu.cn\\
Department of Electronic Engineering,
Tsinghua University, Beijing, P.R. China.  }
\maketitle

\begin{abstract}
Fano's inequality reveals the relation between the conditional entropy and the probability of error . It has been the key tool in proving the converse of coding theorems in the past sixty years. In this paper, an extended Fano's inequality is proposed, which is tighter and more applicable for codings in the finite blocklength regime. Lower bounds on the mutual information and an upper bound on the codebook size are also given, which are shown to be tighter than the original Fano's inequality. Especially, the extended Fano's inequality is tight for some symmetric channels such as the $q$-ary symmetric channels (QSC).
\end{abstract}


\section{Introduction}
As known to all, Shannon's information theory deals mainly with the representation and transmission of information. In the development of both source and channel coding theorems, especially for their converses, Fano's inequality serves as the key tool \cite{Fano1952}.

\thm{Fano's Inequality}

$X$ and $Y$ are two random variables following $(X,Y) \backsim$ $p(x,y)$ and $X,Y\in \mathcal{X}$. Define $P_e=\Pr\{X\neq Y\}$, then
\begin{equation}\label{eq:fano}
H(X|Y)\leq H(P_e)+P_e \log(M-1)
\end{equation}
where $M=|\mathcal X|$ is the cardinality of $X$ and $Y$, $H(P_e)$ is the binary entropy function $H(x)=-[x\log x+(1-x)\log(1-x)]$ for $0\leq x\leq 1$. Thus Fano' inequality can be further relaxed by $H(P_e)\leq 1$.

Usually, the left hand side of (\ref{eq:fano}) is referred to as the equivocation, which is quantified by the conditional entropy. Particularly, it represents the uncertainty whether the restored/decoded message $Y$ is the same as the original one, i.e., $X$. On the other hand, the right hand side implies the reliability of the source/channel coding in terms of a function of error probability $P_e$. It was shown in \cite{Raymond2008} that vanishing equivocation implies vanishing error probability. However, vanishing error probability does not necessarily guarantee a vanishing equivocation, especially for some $X, Y$ of countably infinite alphabet. 

In proving the converse of coding theorems, one wants to find the upper bound on the size of the codebook given arbitrary code length and error probability. The following theorem is an immediate inference of Fano's inequality, simple but useful.

\thm\label{th:fano_IXY} \cite{Hanverdu1994} Suppose $X$ and $Y$ are two random variables that take values on the same finite set with cardinality $M$ and at least one of them is equiprobable. Then the mutual information between them satisfies
\begin{equation}\label {eq:fano_IXY}
    I(X;Y)\geq (1-P_e)\log M - H(P_e),
\end{equation}
where $P_e=\Pr\{X\neq Y\}$ and $H(x)=-[x\log x+(1-x)$ $\log(1-x)]$.

As an inference of this result, the following theorem gives an upper bound on the size of a code as a function of the average error probability.

\thm\label{th:fano_converse} \cite{verdu2010} Every $(M,\epsilon)$-code (average probability of error) for a random transformation $P_{Y|X}$ satisfies
\begin{equation}\label {eq:fano_M}
    \log M\leq \frac{1}{(1-\epsilon)} \sup \limits_X I(X;Y) +  \frac{1}{(1-\epsilon)}H(\epsilon),
\end{equation}
where $\epsilon=\Pr\{X\neq Y\}$, $H(x)=-[x\log x+(1-x)\log(1-x)]$.

Although simple, these two theorems are insightful and easy to compute both in theory and numerically.

The classical coding theorems are mainly based on the asymptotic equipartition property (AEP) and typical/joint-typical decoder \cite{Shannon1948}. However, applications of AEP requires infinite long codewords, where the error probability goes either to 0 or 1 as the code length goes to infinity. Although these coding theorems provide fundamental limits for modern communications, research on the finite blocklength coding schemes are more important in engineering applications. Given the block length, upper bounds on the achievable error probability and the achievable code size were obtained in \cite{verdu2010}. Most importantly, a tight approximation for the achievable maximal rate given the error probability and code length was presented.

In this paper, we consider the entropy of one random variable vector conditioned on another, and the corresponding probability of error in guessing one from the other, by proposing an extended Fano's inequality. The extended Fano's equality has better performance by taking advantage of a more careful consideration on the error patterns. It suits codings in the finite blocklength regime better and is useful in bounding the mutual information between random vectors, and the codebook size given the block length and average symbol error probability constraint.

In the following part of this paper, we present the extended Fano's inequality in Section \ref{sec:gener_fano} first. The lower bounds on the mutual information between two random variable vectors and a upper bound on the codebook size given the block length and error probability are given in Section \ref{sec:converse}. An application of the the obtained result to the $q$-ary symmetric channels (QSC) are presented in Section \ref{sec:app}, which shows that the extended Fano's inequality is tight for such channels. Finally, we concluded the paper in Section \ref{sec:conclusion}. Throughout the paper, vectors indicated by bold.

\section{Fano's Inequality Extension}\label{sec:gener_fano}
Although Fano's inequality has been used widely in the past few years, it can be improved by treating the error events more carefully. In this section, a refinement of Fano's inequality is presented, which is tighter and more applicable for finite blocklength coding design.

\begin{thm}{Fano's Inequality Extension}\label{th:gner_fano}

 Suppose that $\boldsymbol{X}=\{X_1,X_2,\cdots, X_n\}$ and $\boldsymbol{Y}=\{Y_1,Y_2,$ $\cdots, Y_n\}$ are two $n$-dimension random vectors where $X_k$ and $Y_k$ $(k=1,2,\cdots,n)$ take values on the same finite set $\mathcal{X}$ with cardinality $|\mathcal{X}|=q$.
 Then the conditional entropy satisfies {\small
 \begin{equation}\label{rt:gener_fano}
    H(\boldsymbol{X}|\boldsymbol{Y})\leq H(\boldsymbol{p})+\sum_{k=1}^n p_k\log\left(C_n^k (q-1)^k\right),
 \end{equation} }
where $H(\boldsymbol{p})=-\sum_{k=0}^np_k\log p_k$ is the discrete entropy function. $\boldsymbol{p}=\{p_0,p_1,\cdots,p_n\}$ is the error distribution, where the error probabilities are $p_k=\Pr\left(H_d(\boldsymbol{X},\boldsymbol{Y})=k\right)$ for $k=0,1,\cdots, n$. $H_d(\boldsymbol{X},\boldsymbol{Y})$ is the generalized Hamming distance, defined as the number of symbols in $\boldsymbol{X}$ that are different from the corresponding symbol in $\boldsymbol{Y}$.
\end{thm}

\begin{proof} Define the error random variable as  $E=k$ if $H_d(\boldsymbol{X},\boldsymbol{Y})=k$ for $k=0,1,\cdots,n$. 
According to the chain rule of the joint entropy, $H(E,\boldsymbol{X}|\boldsymbol{Y})$ can be expressed in the following ways,
\begin{subequations}\label{dr:fano_g}
    \begin{align}
            H(E,\boldsymbol{X}|\boldsymbol{Y})&=H(\boldsymbol{X}|\boldsymbol{Y})+H(E|\boldsymbol{X},\boldsymbol{Y})\\
            &=H(E|\boldsymbol{Y})+H(\boldsymbol{X}|E,\boldsymbol{Y})
    \end{align}
\end{subequations}

Particularly, in (\ref{dr:fano_g}.a), it is clear that $H(E|\boldsymbol{X},\boldsymbol{Y})=0$. Then we have
\begin{equation}\label{dr:Hx_y}
    \begin{split}
        H(\boldsymbol{X}|\boldsymbol{Y})=&H(E|\boldsymbol{Y})+H(\boldsymbol{X}|E,\boldsymbol{Y})\\
        \stackrel{(a)}{\leq} & H(E)+H(\boldsymbol{X}|E,\boldsymbol{Y}),
    \end{split}
\end{equation}
where (a) follows the fact that entropy increases if its condition is removed, i.e., $H(E)\geq H(E|\boldsymbol{Y})$. Particularly, we have $H(E)=H(\boldsymbol{p})=-\sum_{k=0}^np_k\log p_k$.

According to its definition, we have
 \begin{equation}\label{dr:condi_err}
    H(\boldsymbol{X}|E,\boldsymbol{Y})=\sum_{k=0}^n p_k H(\boldsymbol{X}|E=k,\boldsymbol{Y}).
 \end{equation}

 When considering $H(\boldsymbol{X}|E=k,\boldsymbol{Y})$, we know that there are $k$ disaccord symbol pairs between $\boldsymbol{X}$ and $\boldsymbol{Y}$. For each fixed $\boldsymbol{Y}=\boldsymbol{y}$, every symbol in $\boldsymbol{X}$ which belongs to a disaccord pair has $q-1$ possible choices except the one in $\boldsymbol{y}$. Thus $\boldsymbol{X}|(E=k,\boldsymbol{y})$ has $(q-1)^k$ choices. Besides, there $C_n^k$ selections for the positions of error symbols for each given $k$. Therefore, the total number of possible codeword $\boldsymbol{X}$ is $C_n^k (q-1)^k$, which means
\begin{equation}
    H(\boldsymbol{X}|E=k,\boldsymbol{Y})\leq \log\left( C_n^k(q-1)^k\right).
\end{equation}

Particularly, note that $H(\boldsymbol{X}|E=0,\boldsymbol{Y})=0$ since there is no uncertainty in determining $\boldsymbol{X}$ from $\boldsymbol{Y}$, if they are the same.

Then, (\ref{dr:condi_err}) can be written as
\begin{equation}\label{Hey}
    \begin{split}
        H(\boldsymbol{X}|E,\boldsymbol{Y})=&\sum_{k=1}^n P(E=k) H(\boldsymbol{X}|E=k,\boldsymbol{Y}) \\
        \leq &\sum_{k=1}^n p_k \log\left(C_n^k (q-1)^k\right).
    \end{split}
\end{equation}

By combining (\ref{dr:Hx_y}) and (\ref{Hey}), the proof of the theorem is completed.

 \end{proof}

\rem In fact, the error distribution $\boldsymbol{p}$ is easy to calculate, especially for some special channels. For example, the discrete $q$-ary symmetric channels is shown in Fig \ref{fig:qbc}.  In this situation, $p_k=C_n^k ((q-1)\varepsilon)^k (1-(q-1)\varepsilon)^{n-k}$.

\begin{figure}[!t]
\centering
\includegraphics[width=2.5in]{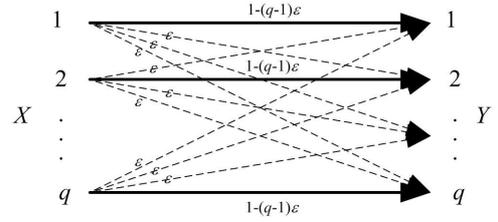}
\caption{The $q$-ary Symmetric Channel} \label{fig:qbc}
\end{figure}

\rem It is clear that Theorem \ref{th:gner_fano} is a generalization of Fano's inequality. Specifically, when the block length is 1, i.e., $n=1$, we have $p_1=p_e$, $p_0=1-p_e$, and $C_1^1=1$. In this case, Theorem \ref{th:gner_fano} reduces to
\begin{equation}
    H(X|Y)\leq H(p_e)+p_e\log(q-1)
\end{equation}
which is exactly the same as Fano's inequality.

As a variant of Theorem \ref{th:gner_fano}, the following theorem presents the conditional entropy in terms of relative entropy.

\begin{thm}\label{th:gner_fanoRela}
Suppose that $\boldsymbol{X}=\{X_1,X_2,\cdots, X_n\}$ and $\boldsymbol{Y}=\{Y_1,Y_2,\cdots, Y_n\}$ are two $n$-dimension random vectors where $X_k$ and $Y_k$ $(k=1,2,\cdots,n)$ take values on the same finite set $\mathcal{X}$ with cardinality $|\mathcal{X}|=q$.
 Then the conditional entropy satisfies
 \begin{equation}
    H(\boldsymbol{X}|\boldsymbol{Y})\leq n \log q- D(\boldsymbol{p}\|\boldsymbol{q}),
 \end{equation}
where $D(\boldsymbol{p}\|\boldsymbol{q})=\sum_{k=0}^n p_k \log\frac{p_k}{q_k}$ is the discrete relative entropy function. The error probabilities are $p_k=\Pr\left(H_d(\boldsymbol{X},\boldsymbol{Y})=k\right)$ for $k=0,1,\cdots, n$. Donate $\boldsymbol{p}=\{p_0,p_1,\cdots,p_n\}$, and $\boldsymbol{q}=\{q_0,q_1,\cdots,q_n\}$ is a probability distribution with $q_k=\frac{C_n^k (q-1)^k}{q^n}$
\end{thm}

\begin{proof}
Firstly, we know from the binomial theorem that $\{ \frac{C_n^k a^kb^{n-k}}{(a+b)^n} \}_{0\leq n}$ is a probability distribution. Let $a=q-1$ and $b=1$, we know that $\boldsymbol{q}$ is a probability distribution where $q_k=\frac{C_n^k (q-1)^k}{q^n}$.

According to Theorem \ref{th:gner_fano},
\begin{equation}\nonumber
    H(\boldsymbol{X}|\boldsymbol{Y})
     \stackrel{(a)}{\leq}-\sum_{k=0}^n p_k \log\frac{p_k}{C_n^k (q-1)^k} \nonumber 
    = n \log q -D(\boldsymbol{p}\parallel \boldsymbol{q})
\end{equation}
where (a) holds because $\log \left(C_n^k (q-1)^k \right)=0$ for $k=0$.
\end{proof}

\rem By the definition of $\boldsymbol{p}$, it is clear that it reflects the error performance of the channel and is totaly determined by the channel itself. On the contrary, $\boldsymbol{q}$ is a distribution where each error patten is assumed to appear equiprobablely. In this situation, the probability that there are $k$ error symbols  in the codeword is $q_k=C_n^k (\frac{q-1}{q})^k(\frac{1}{q})^{n-k}$. Thus $D(\boldsymbol{p}\|\boldsymbol{q})$ is the distance between the actual error pattern distribution and the uniform error pattern distribution. Particularly, if the channel is an error free one, i.e., $p_0=1$ and $p_k=0$ for $1\leq k\leq n$, we have $D(\boldsymbol{p}\|\boldsymbol{q})=\sum_{k=0}^n p_k \log\frac{p_k}{q_k}=1\cdot\log \frac{1}{q_0}=\log q^n$. According to Theorem \ref{th:gner_fanoRela}, we get $H(\boldsymbol{X}|\boldsymbol{Y})=0$, which is reasonable with the assumption on the channel. In this sense, Theorem \ref{th:gner_fanoRela} is tight.

The Fano's inequality has been playing an important role in the history of information theory because it built a close connection between conditional entropy and error probability. For extended Fano's inequality given in Theorem \ref{th:gner_fano}, it is especially applicable in finite blocklength coding. It also presents the relationship between conditional entropy and error probabilities, which are defined as follows.

\begin{defn}
    \textit{Block error probability} $P_b$ is the average error probability of a block (codeword), i.e., $P_b=\Pr\{\boldsymbol{X}\neq\boldsymbol{Y}\}$. Then we have
    $P_b=\sum_{k=1}^n p_k$. Thus, we have $P_b\geq p_k$ for any $0<k\leq n$.
\end{defn}

\begin{defn}
    \textit{Symbol error probability} $P_s$ is the average error probability of a symbol, i.e., $P_s=\Pr\{X_k\neq Y_k\}$, which can be expressed by $P_s=\frac{1}{n}\sum_{k=1}^n kp_k$.
\end{defn}

\rem In many communication systems, especially those using error correction channel codings, a block error doesn't imply a system failure. On the contrary, the error can be corrected or part of the block can still be used with some performance degradation. In this case, the symbol error is more useful than the block error.

Particularly, a corollary following our result as shown below will answer this problem.

\begin{cor}\label{cor:1}
    Suppose that $\boldsymbol{X}=\{X_1,X_2,\cdots, X_n\}$ and $\boldsymbol{Y}=\{Y_1,Y_2,\cdots, Y_n\}$ are two $n$-dimension random vectors where $X_k$ and $Y_k$ take values on the same finite set $\mathcal{X}$ with cardinality $|\mathcal{X}|=q$.
 Then the conditional entropy satisfies
 \begin{equation}\label{rt:cor_1}
    H(\boldsymbol{X}|\boldsymbol{Y}) \leq n- D(\boldsymbol{p} \| \boldsymbol{w}) + n P_s\log (q-1),
 \end{equation}
where $\boldsymbol{w}=\{w_0,w_1,\cdots,w_n\}$ is a probability distribution with $w_k=\frac{C_n^k}{2^n}$.
\end{cor}

\begin{proof}
According to Theorem \ref{th:gner_fano}, one has{\small
\begin{equation}\nonumber
    \begin{split}
        &H(\boldsymbol{X}|\boldsymbol{Y})\leq H(\boldsymbol{p}) + \sum_{k=0}^n p_k\log\left(\frac{C_n^k}{2^n}2^n \right) + \sum_{k=1}^n p_k\log (q-1)^k \\
        &=-\sum_0^n p_k \log p_k + \sum_{k=0}^n p_k\log w_k+\sum_{k=0}^n p_k \log 2^n \\ &=n- D(\boldsymbol{p} \| \boldsymbol{w}) + n P_s\log (q-1).
    \end{split}
\end{equation} }
\end{proof}

\rem Since the distribution $\boldsymbol{w}$ can be expressed as $w_k=C_n^k (\frac12)^k(\frac12)^{n-k}$, which is a binomial distribution with the symbol error probability of $0.5$.  $D(\boldsymbol{p} \| \boldsymbol{w})$ is a measure of the distance between the error probability distribution and the binomial distribution with parameter $0.5$.

\rem If one takes $n=1$, Corollary \ref{cor:1} will reduces to $H(X|Y)\leq1+P_s\log (q-1)$, which is a frequently used form of Fano's inequality.

\rem It is seen that the extended Fano's inequality builds a natural connection between conditional entropy and symbol error and is especially applicable for finite length codings.

\section{Converse Results}\label{sec:converse}

\subsection{Lower Bounds on the Mutual Information}
Based on the proposed generalized Fano's inequality, the following lower bounds on the mutual information between $\boldsymbol{X}$ and $\boldsymbol{Y}$ can be obtained.

\begin{thm}\label{th:gner_fano_IXY}
 Suppose that $\boldsymbol{X}=\{X_1,X_2,\cdots, X_n\}$ and $\boldsymbol{Y}=\{Y_1,Y_2,\cdots, Y_n\}$ are two $n$-dimension random vectors that satisfy the following.
 \begin{enumerate}
   \item $X_k$ and $Y_k$ $(k=1,2,\cdots,n)$ take values on the same finite set $\mathcal{X}$ with cardinality $|\mathcal{X}|=q$.
   \item Either $\boldsymbol{X}$ or $\boldsymbol{Y}$ is equiprobable.
   \item The error probabilities are $p_k=\Pr\left(H_d(\boldsymbol{X},\boldsymbol{Y})=k\right)$ for $k=0,1,\cdots, n$. Donate the error distribution as $\boldsymbol{p}=\{p_0,p_1,\cdots,p_n\}$.
 \end{enumerate}
Then the mutual information between $\boldsymbol{X}$ and $\boldsymbol{Y}$ satisfies
 \begin{equation}
    I(\boldsymbol{X};\boldsymbol{Y})\geq n\log q - H(\boldsymbol{p}) - \sum_{k=1}^n p_k\log\left(C_n^k (q-1)^k\right),
 \end{equation}
where $H(\boldsymbol{p})=-\sum_{k=0}^np_k\log p_k$ is the entropy function.
\end{thm}

\begin{proof}
If $\boldsymbol{X}$ is equiprobable, $H(\boldsymbol{X})=\log q^n=n\log q$.

On the other hand, the mutual information is given by $I(\boldsymbol{X};\boldsymbol{Y})=H(\boldsymbol{X})-H(\boldsymbol{X}|\boldsymbol{Y})$. Together with Theorem \ref{th:gner_fano},
 \begin{equation}\label{dr:gener_IXY1}
    I(\boldsymbol{X};\boldsymbol{Y})\geq n\log q - H(\boldsymbol{p}) - \sum_{k=1}^n p_k\log\left(C_n^k (q-1)^k\right).
 \end{equation}

Note that $\boldsymbol{X}$ and $\boldsymbol{Y}$ are totally symmetric in (\ref{dr:gener_IXY1}). Therefore, if $\boldsymbol{Y}$ is assumed to be equiprobable at the beginning of the proof, one can get the same result. Thus Theorem \ref{th:gner_fano_IXY} is proved.
\end{proof}

By using Theorem \ref{th:gner_fanoRela}, the mutual information between $\boldsymbol{X}$ and $\boldsymbol{Y}$ can be bounded by the following Corollary.

\begin{cor}\label{cor:4}
Suppose that $\boldsymbol{X}=\{X_1,X_2,\cdots, X_n\}$ and $\boldsymbol{Y}=\{Y_1,Y_2,\cdots, Y_n\}$ are two $n$-dimension random vectors that satisfy the following.
 \begin{enumerate}
   \item $X_k$ and $Y_k$ $(k=1,2,\cdots,n)$ take values on the same finite set $\mathcal{X}$ with cardinality $|\mathcal{X}|=q$.
   \item Either $\boldsymbol{X}$ or $\boldsymbol{Y}$ is equiprobable.
   \item The error probabilities are $p_k=\Pr\left(H_d(\boldsymbol{X},\boldsymbol{Y})=k\right)$ and $\boldsymbol{p}=\{p_0,p_1,\cdots,p_n\}$.
 \end{enumerate}
Then the mutual information between $\boldsymbol{X}$ and $\boldsymbol{Y}$ satisfies
 \begin{equation}\label{rt:cor4}
     I(\boldsymbol{X};\boldsymbol{Y})\geq D(\boldsymbol{p}\|\boldsymbol{q}) \end{equation}
 where $D(\boldsymbol{p}\|\boldsymbol{q})=\sum_{k=0}^n p_k \log\frac{p_k}{q_k}$ is the discrete relative entropy function and $\boldsymbol{q}=\{q_0,q_1,\cdots,q_n\}$ is a probability distribution with $q_k=\frac{C_n^k (q-1)^k}{q^n}$
\end{cor}

The distribution $\boldsymbol{q}$ means that the symbol in $\boldsymbol{Y}$ takes any value on $\mathcal{B}$ with equal probability, regardless of what is sent in $\boldsymbol{X}$. So it is a pure random distribution when $X$ and $Y$ are independent from each other. The most desirable coding scheme is that its error distribution $\boldsymbol{p}$ is farthermost from $\boldsymbol{q}$, which also ensures a larger coding rate.

\subsection{Upper Bounds on the Codebook Size}
Suppose $\mathcal{X}$ is a finite alphabet with cardinality $|\mathcal{X}|=q$. Let's consider the input and output alphabets $\mathcal{A}^n=\mathcal{B}^n \subseteq \mathcal{X}^n$ with $|\mathcal{A}^n|=|\mathcal{B}^n|=M$ and a channel to be a sequence of conditional probabilities \cite{Hoverdu2010} $\{P_{\boldsymbol{Y}|\boldsymbol{X}} :\mathcal{A}^n \mapsto \mathcal{B}^n\}$. We donate a codebook with $M$ codewords by $(\boldsymbol{X}_1,\boldsymbol{X}_2,\cdots,\boldsymbol{X}_M)\in \mathcal{A}^n$. A decoder is a random transformation $P_{Z|\boldsymbol{Y}}:\mathcal{B}^n \mapsto \{0,1,2,\cdots,M\}$ where $0$ indicates that the decoder choose error. If messages are equiprobable, the average error probability is defined as $P_b=1-\frac{1}{M} P_{Z|\boldsymbol{X}} (m|\boldsymbol{X}_m)$.
An codebook with $M$ codewords and a decoder whose average probability of error is smaller than $\epsilon$ are called an $(n, M, \epsilon)$-code.

An upper bound on the size of a code as a function of the average probability of symbol error follows the Corollary \ref{cor:1}.

\begin{thm}\label{th:conver_1}
Every $(n, M, \epsilon)$-code for a random transformation $P_{Z|\boldsymbol{Y}}$ satisfies
\begin{equation}\label{rt:conver_1}
    \log M \leq \sup\limits_{\boldsymbol{X}} I(\boldsymbol{X};\boldsymbol{Y}) - D(\boldsymbol{p} \| \boldsymbol{w}) + n\left(1 + P_s\log (q-1)\right)
\end{equation}
where $\boldsymbol{p}=\{p_0,p_1,\cdots,p_n\}$ is the error distribution with $p_k=\Pr\left(H_d(\boldsymbol{X},\boldsymbol{Y})=k\right)$ for $k=0,1,\cdots, n$, $\boldsymbol{w}=\{w_0,w_1,$ $\cdots, w_n\}$ is a probability distribution with $w_k=\frac{C_n^k}{2^n}$.
\end{thm}

\begin{proof}

Since the messages are equiprobable, we have $H(\boldsymbol{X})=\log M$. According Corollary \ref{cor:1},
\begin{equation}\label{dr:th_con1}
    \begin{split}
        I(\boldsymbol{X};\boldsymbol{Y})=&H(\boldsymbol{X})-H(\boldsymbol{X}|\boldsymbol{Y})\\
        \geq &\log M - n + D(\boldsymbol{p} \| \boldsymbol{w}) - n P_s\log (q-1).
    \end{split}
\end{equation}

Solving $\log M$ from (\ref{dr:th_con1}), one can get (\ref{rt:conver_1}), which completes the proof.
\end{proof}

\section{Application to Channel Coding}\label{sec:app}
Consider information transmission over a memoryless discrete $q$-ary symmetric channel  with a channel code $(n,M,\epsilon)$ with crossover probability $\varepsilon$, as shown in Fig. \ref{fig:qbc}. In this case, the probability of symbol error is $p_e=(q-1)\varepsilon$.

Then the error probabilities are
\begin{equation}
    p_k=\Pr\{H_d(\boldsymbol{X},\boldsymbol{Y})=k\}
    =C_n^k p_e^k (1-p_e)^{n-k}
\end{equation}
and the block error probability is
\begin{equation}\label{eq:pb}
    P_b=1-p_0=1-(1-p_e)^n.
\end{equation}

Using the extended Fano's inequality in Theorem \ref{th:gner_fano}, we have
\begin{equation}\label{sim:hexy}
    H_e(\boldsymbol{X}|\boldsymbol{Y})\leq\sum_{k=0}^n p_k\log \left( \frac{C_n^k (q-1)^k} {p_k} \right).
\end{equation}

It is easy to see that the conditional entropy in theory is
\begin{equation}\label{sim:hxy}
    \begin{split}
    &H(\boldsymbol{X}|\boldsymbol{Y})=H\left(1-p_e,\varepsilon,\cdots,\varepsilon\right)\\
    =&-(1-p_e)\log(1-p_e)-(q-1)\log\varepsilon.
    \end{split}
\end{equation}

By Corollary \ref{cor:4}, mutual information is lower bounded by{\small
\begin{equation}\label{sim:iexy}
    \begin{split}
        &I_e(\boldsymbol{X};\boldsymbol{Y})
        \geq\sum_{k=0}^n p_k\log \left( \frac{C_n^k p_e^k (1-p_e)^{n-k} } {(C_n^k (q-1)^k)/q^n} \right)\\
        &=n\log q+ \sum_{k=0}^n p_k \left[ k\log p_e+(n-k)\log\left(1-p_e\right)\right.\\
        &~~~~~~~~~~~~~~~~~\left.-k\log(q-1) \right].
    \end{split}
\end{equation} }
while the capacity of the memoryless QSC is given by
\begin{equation}\label{sim:ixy}
    \begin{split}
    &I(X;Y)=\log q- H\left(1-p_e,\varepsilon,\cdots,\varepsilon\right)\\
    =&\log q + (1-p_e)\log(1-p_e) + (q-1)\varepsilon\log\varepsilon.
    \end{split}
\end{equation}

And the relative entropy $D(\boldsymbol{p}\|\boldsymbol{w})$ can be derived as
\begin{equation}
    \begin{split}
        &D(\boldsymbol{p}\|\boldsymbol{w})=\sum_{k=0}^n p_k\log \left( \frac{C_n^k \varepsilon^k (1-(q-1)\varepsilon)^{n-k} } {C_n^k /2^n} \right)\\
        =&n+ \sum_{k=0}^n p_k \left[ k\log p_e+(n-k)\log(1-p_e)\right]\\
    \end{split}
\end{equation}

For a given average symbol error probability constraint $P_s=\epsilon$, the upper bound on the maximum codebook size given by Theorem \ref{th:conver_1} is {\small
\begin{equation}\label{sim:mexy}
    \begin{split}
    \log M_e\leq & I(\boldsymbol{X};\boldsymbol{Y}) - D(\boldsymbol{p}\|\boldsymbol{w})
    + n\left(1+P_s\log(q-1)\right)\\
    =&n\log q-nH(1-p_e,\varepsilon,\cdots,\varepsilon)-\sum_{k=0}^n p_k\left[k\log p_e\right.\\ 
    +&\left.(n-k)\log(1-p_e)\right] + n\epsilon\log(q-1).
    \end{split}
\end{equation} }

On the other hand, by Fano's inequality we have
\begin{equation}\label{sim:hfxy}
    H_f(\boldsymbol{X}|\boldsymbol{Y})\leq H(P_b)+P_b\log(q^n-1)
\end{equation}
with $P_b$ given by (\ref{eq:pb}).

Then the lower bound of the mutual information is
\begin{equation}\label{sim:ifxy}
    I_f(\boldsymbol{X};\boldsymbol{Y})\geq (1-P_b)\log q^n-H(P_b).
\end{equation}

Finally, the upper bound on the codebook size is
\begin{equation}\label{sim:mxy}
    \log M_f \leq \frac{1}{1-\epsilon} \left(nI(X;Y)+H(\epsilon)\right).
\end{equation}

Suppose the QSC parameter are $\varepsilon=0.001$ and $q=7$, we calculated the bounds on conditional entropy, mutual information and codebook size by our proposed results and Fano's inequality.

Firstly, the upper bound on the conditional entropy is presented in Fig. \ref{fig:hxy}. Specially, $H_e(\boldsymbol{X}|\boldsymbol{Y})$ is obtained by the extended Fano's inequality (\ref{sim:hexy}), $H_f(\boldsymbol{X}|\boldsymbol{Y})$ is calculated according to Fano's inequality (\ref{sim:hfxy}) and $H(\boldsymbol{X}|\boldsymbol{Y})$ is the conditional entropy in theory (\ref{sim:hxy}). It is clear that Theorem \ref{th:gner_fano} is tighter than Fano's inequality. Particularly, we have $H_e(\boldsymbol{X}|\boldsymbol{Y})= H(\boldsymbol{X}|\boldsymbol{Y})$ for the QSC. This is because the error distributions are the same for any $\boldsymbol{Y}=\boldsymbol{y}$. So $H(E)=H(E|\boldsymbol{Y})$ holds. Besides, the error pattern is uniformly distributed for a given $k$, regardless of $\boldsymbol{y}$ and $H(E)=H(\boldsymbol{Y}| E,\boldsymbol{Y})=H(\boldsymbol{Y}| E)=\log C_n^k(q-1)^k$ holds. Therefore, the upper bound is tight. However, for Fano's equality, there are relaxations in both $H(E)=H(E|\boldsymbol{Y})$ to $H(E)$ and $H(\boldsymbol{X}|E,\boldsymbol{Y})$ to $P_b\log (M-1)$.
\begin{figure}[!t]
\centering
\includegraphics[width=2.7in]{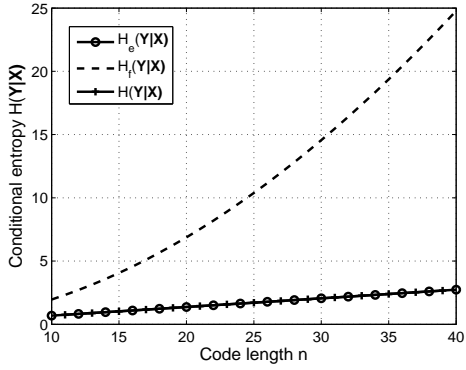}
\caption{The upper bound on conditional entropy, $q=7,\varepsilon=0.001$} \label{fig:hxy}
\end{figure}
\begin{figure}[!t]
\centering
\includegraphics[width=2.9in]{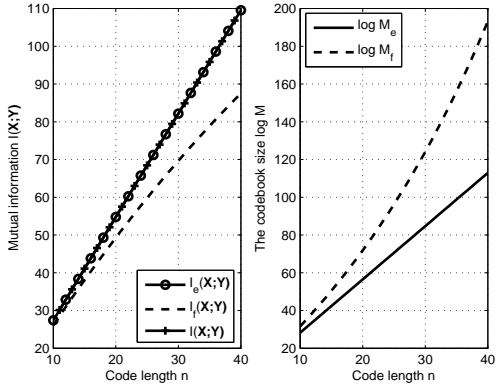}
\caption{Bounds on mutual information and codebook size v.s. blockleng, $q=7,\varepsilon=0.001$} \label{fig:ixy}
\end{figure}

Similarly, the lower bound on mutual information $I_e(\boldsymbol{X};\boldsymbol{Y})$ given by (\ref{sim:iexy}) coincides with $I(\boldsymbol{X};\boldsymbol{Y})$ in theory, given by (\ref{sim:ixy}) and is better than that given by Fano's inequality (\ref{sim:ifxy}).

When we use the upper bound on the codebook size in Theorem \ref{th:conver_1}, it should be noted that it is presented as a function of symbol error probability $P_s$. In fact, $P_b$ is always larger than $P_s$. Therefore, we use the same fraction of them in the calculation of the bounds to make sense of the comparison, i.e., $\epsilon=\frac{P_s}{2}$ for (\ref{sim:mexy}) and $\epsilon=\frac{P_b}{2}$ for (\ref{sim:mxy}). It is also seen from Fig. \ref{fig:ixy} that our new developed result is tighter.

The performances of Theorem \ref{th:gner_fano_IXY} and Theorem \ref{th:conver_1} versus the QSC parameter $\varepsilon$ are shown in Fig. \ref{fig:epslong}, where the block length is chosen as $n=30$. As shown, the mutual information bound is tight and our results are much better than Fano's inequality. In the calculation of the upper bounds on codebook size, the selection of the error probability constraints are also chosen as $\epsilon_e=\frac{P_s(\varepsilon)}{2}$ and $\epsilon_f=\frac{P_b(\varepsilon)}{2}$ so that they are comparable.
\begin{figure}[!t]
\centering
\includegraphics[width=2.9in]{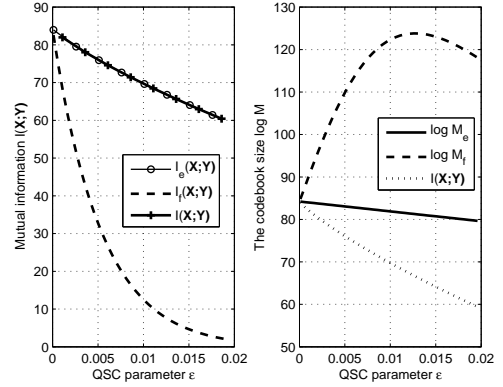}
\caption{Bounds on mutual information and codebook size v.s. QSC parameter $\varepsilon$, $q=7,n=30$} \label{fig:epslong}
\end{figure}

\section{Conclusion}\label{sec:conclusion}
In this paper, we revisited Fano's inequality and extended  it to a general form. Particularly, the relation between the conditional entropy and error probability of two random vectors was considered, other than that between two random variables. This makes the developed results more suitable for source/channel codings in the finite blocklength regime. By investigating the block error pattern more detailedly, the conditional entropy of the original random vector given the received one is upper bounded more tightly by the extended Fano's inequality. Furthermore, the extended Fano's inequality is completely tight for some symmetric channels such the $q$-ary symmetric channels. Converse results are also presented in terms of lower bounds on the mutual information and a upper bound on the codebook size under the blocklength and symbol error constraints, which also have better performances.

\section*{Acknowledgement}
This work was partially supported by INC research grant of Chinese University of Hongkong and the China Major State Basic Research Development Program (973 Program) No. 2012CB316100(2).

\end{document}